\documentclass{article}
\usepackage{amssymb,amsbsy,amsmath,amsfonts,amssymb,amscd,amsthm}
\usepackage{graphicx,multirow,amsmath}
\usepackage{subfig}

\newtheorem{prop}{Proposition}

\newcommand{\thre}{V^\ast}
\newcommand{\optipri}{p^{(\rho,\gamma)}}

\begin{document}

\title{Stock loans in incomplete markets}

\author{ Matheus R Grasselli\footnotemark[1] \  \and Cesar G. Velez
\footnotemark[2]}

\renewcommand{\thefootnote}{\fnsymbol{footnote}} \footnotetext[1]{%
McMaster University} 
\footnotetext[2]{%
Universidad Nacional de Colombia } 


\maketitle

\begin{abstract}
A stock loan is a contract whereby a stockholder uses shares as collateral to borrow money from a bank or financial institution. In Xia and Zhou (2007), this contract is modeled as a perpetual American option with a time varying strike and analyzed in detail within a risk--neutral framework. In this paper, we extend the valuation of such loans to an incomplete market setting, which takes into account the natural trading restrictions faced by the client. When the maturity of the loan is infinite, we use a time--homogeneous utility maximization problem to obtain an exact formula for the value of the loan fee to be charged by the bank. For loans of finite maturity, we characterize the fee using variational inequality techniques. In both cases we show analytically how the fee varies with the model parameters and illustrate the results numerically.
\end{abstract}


\section{Introduction}

A stock loan is a contract between two parties: the lender, usually a bank or other financial institution providing a loan, and the borrower, represented by a client who owns one share of a stock used as collateral for the loan. Several reasons might motivate the client to get into such a deal. For example he might not want to sell his stock or even face selling restrictions, while at the same time being in need of available funds to attend to another financial operation. 

Our main task consists of determining the fair values of the parameters of the loan, particularly the value of the fee that the bank
charges for the service along with the interest rate to be charge over the amount borrowed, taking into account the stock price at the moment of 
taking the loan. In addition, we take into account the fact that the bank typically collects any dividends paid by the stock for the duration of the loan. Finally, whereas the client can recover the stock at any time by paying the loan principal plus interest, he is not obliged to do so, even if the stock price falls down, and this optionality also needs to be accounted in the valuation of the loan. 

In \cite{XiaZhou07}, a stock loan is modeled as a perpetual American option with a time varying strike and analyzed in detail using probabilistic methods within the Black-Scholes framework. Assuming that the risk neutral dynamics of the stock follows a geometric Brownian motion, they obtained explicit formulas for the bank's fee in terms of the amount lent and the stock price at the moment of signing the loan. Implicit in their use of the risk neutral paradigm is the assumption that the option can be replicated by trading in the underlying stock and the money market. Whereas this is certainly plausible from the bank's point of view, we argue that neither type of trade is readily available for the client, who presumably does not have unrestricted access to the money market (hence the need to post collateral in the form of a stock) nor can freely trade in the stock (otherwise he would simply sell the stock instead of take the loan). Moreover, while risk neutral valuation yields the fair price at which the option itself can be traded in the market without introducing arbitrage opportunities, a stock loan typically cannot be sold or bought in a secondary market once it is initiated. In other words, the client does not operate in the frictionless market that is assumed by the Black--Scholes framework.  
 
Accordingly, we treat a stock loan as an option in an incomplete market. We assume that the client cannot trade directly in the underlying stock, but is allowed to trade in a portfolio of assets that is imperfectly correlated to the stock. In this way, since the client cannot perfectly hedge the embedded optionality, he faces some non-diversifiable risk for the duration of the loan. We assume that the client is a risk averse economic agent and model his preferences by an exponential utility function. We then use utility indifference arguments to value the stock loan from the point of view of the client both for infinite maturity, where semi-explicit formulas are still available, and for finite
maturity, where numerical computations are needed. Finally we assume that the bank is well diversified and relate the fee charged by the bank
with the hedging cost for a barrier-type option reflecting the exercise behavior of the client.

Although we present the analysis using a financial asset as the collateral, it is clear that the same results can be applied to loans against other types of assets, such as real estate or inventories, provided their value is observable and follows a dynamics that can be modeled according to \eqref{market}.

\section{Model set up}

We consider a market consisting of two correlated assets $S$ and $V$ with {\em discounted} prices given by
\begin{equation}
\label{market}
\begin{aligned}
dS_{t} & =(\mu_{1}-r) S_{t}dt+\sigma_{1} S_{t}dW^1_{t} \\
dV_{t} & =(\mu_{2}-r) V_{t}dt +\sigma_{2}V_{t}(\rho dW^1_t+\sqrt{1-\rho^{2}}dW^2_t),
\end{aligned}
\end{equation}
for $t_0\leq t\leq T\leq \infty$, where $W=(W^1,W^2)$ is a standard two--dimensional Brownian motion. We suppose further that the client can trade dynamically by holding $H_t$ units of the asset $S_t$ and investing the remaining of his wealth in a bank account with normalized value $B_t=e^{r(t-t_0)}$ for a constant interest rate $r$. It follows that the discounted value of the corresponding self--financing portfolio satisfies 
\begin{equation}
dX^\pi_t= \pi_t(\mu_1-r)dt+\pi_t\sigma_1dW^1_t, \quad t_0\leq t\leq T,
\label{wealth}
\end{equation}
where $\pi_t=H_tS_t$.

At a given time $t_0$, the client borrows an amount $L$ from a bank leaving the asset with value $V_{t_0}$ as a collateral. We assume that the bank collects the dividends paid by the underlying asset $V$ at a rate $\delta$ for the duration of the loan. In addition, the bank charges the client a fee $c$ and stipulates an interest rate $\alpha$ to be charged on the loan amount $L$, so that the client can redeem the asset with value $e^{r(t-t_0)}V_t$ at time $t_0\leq t\leq T$ by paying an amount $e^{\alpha (t-t_0)}L$.  At the maturity time $T$, we assume that the client needs to decide between repaying the loan or forfeiting the underlying asset indefinitely.

In other words, at the beginning of the loan the client gives the bank an asset worth $V_{t_0}$ and receives a net amount $(L-c)$ plus the option to buy back an asset with market price $e^{r(t-t_0)}V_t$ for an amount $e^{\alpha (t-t_0)}L$. Denoting the cost of this option for the bank by $C_{t_0}$, the loan parameters are related by 
\begin{equation}
c=L+C_{t_0}-V_{t_0}
\label{relation}
\end{equation}

\section{Infinite maturity}
\label{infinite}

Let us first assume that $T=\infty$ and that $\alpha=r$. Given an exponential utility $U(x)=-e^{-\gamma x}$, we consider a client trying to maximize the expected utility of discounted wealth. We assume that, upon repaying the loan at time $\tau$, the borrower adds the discounted payoff $(V_\tau-L)$ to his discounted wealth $X^\pi_\tau$ and continues to invest optimally. Accordingly, having taken the loan at time $t_{0}$, the borrower needs to solve the following optimization problem:
\begin{equation}\label{opti1}
 G(x,v)=\sup_{(\tau,\pi)\in\mathcal{A}}\mathbb{E}\big[-e^{\frac{(\mu_1-r)^2}{2\sigma^2}\tau}e^{-\gamma(X^{\pi}_{\tau}+(V_{\tau}-L)^{+})}\big| X^\pi_{t_{0}}=x,V_{t_{0}}=v\big].
\end{equation}
Here $\mathcal{A}$ is a set of admissible pairs $(\tau,\pi)$, where $\tau\in[0,\infty]$ is a stopping time and $\pi$ is a portfolio process. Observe that the factor $e^{\frac{(\mu_1-r)^2}{2\sigma^2}\tau}$ leads to a horizon unbiased optimization problem (see Appendix 1 of \cite{Henderson07} for details), while the choice $\alpha=r$ removes the dependence on time from the factor $e^{(\alpha-r)(t-t_{0})}L$ at the repayment date. Combined with the infinite maturity assumption, this allows us to deduce that the borrower should decide to pay back the loan at the first time that $V$ reaches a stationary threshold 
$\thre$, that is
\begin{equation}\tau^\ast=\inf\{s\geq t_0:V_{s}=\thre\}.
\label{tau1}
\end{equation}

We follow  \cite{HodgesNeuberger89} and define the indifference value for the option to pay back the loan as the amount $p(v)$ satisfying 
\begin{equation}
 G(x,0)=G(x-p(v),v).
\end{equation}

The following proposition summarizes the resulst in \cite{Henderson07} regarding the value function $G(x,v)$, the threshold $\thre$ and the indifference value 
$p(v)$.
\begin{prop}[Henderson, 2007] The function $G(x,v)$ solves the following non-linear HJB equation 
\begin{equation}\label{HJB1}
\frac{(\mu_1-r)^2}{2\sigma_1^2}G+\frac{\sigma_{2}^{2}v^2}{2}G_{vv}+(\mu_{2}-r)vG_{v}-\frac{(\rho\sigma_{1}\sigma_{2}vG_{xv}+(\mu_{1}-r)G_{x})^{2}}{\sigma^{2}_{1}G_{xx}}=0
\end{equation}
subject to the following boundary, value matching and smooth pasting conditions:
\begin{align*}
 G(x,0) & = -e^{-\gamma x}\\
G(x,\thre) & =-e^{-\gamma(x+\thre-L)}\\
G_{v}(x,\thre) & = \gamma e^{-\gamma(x+\thre-L)}.
\end{align*}
Let $\beta=1-\frac{2}{\sigma_2}\left(\frac{\mu_2-r}{\sigma_2}-\rho\frac{\mu_1-r}{\sigma_1}\right)$. If $\beta>0$, the threshold $\thre>L$ is the unique solution to 
\begin{equation}
\thre -L=\frac{1}{\gamma(1-\rho^{2})}\log\left[1+\frac{\gamma(1-\rho^{2})\thre}{\beta}\right]
\label{shold}
\end{equation}
and the solution to (\ref{HJB1}) and associated conditions is given by
\begin{equation}
G(x,v)=\left\{\begin{aligned}
& -e^{-\gamma x}\bigg[1-(1-e^{-\gamma(\thre-L)(1-\rho^{2})})\big(\frac{v}{\thre}\big)^{\beta}\bigg]^{\frac{1}{1-\rho^{2}}}, \quad\mbox{if}\quad v<V^\ast \\
&-e^{\gamma x}e^{-\gamma(v-L)}, \quad\mbox{if}\quad v \geq V^\ast.
\end{aligned}\right.
\end{equation} 
In this case, the indifference value $\optipri$ is given by
\begin{equation}
\label{indiff}
 p(v)=\left\{\begin{aligned}
&-\frac{1}{\gamma(1-\rho^{2})}\log\bigg[(e^{-\gamma(\thre-L)(1-\rho^{2})}-1)\big(\frac{v}{\thre}\big)^{\beta}+1\bigg], \quad\mbox{if}\quad v<V^\ast \\
&(v-L), \quad\mbox{if}\quad v \geq V^\ast. \end{aligned}\right.
\end{equation}
Alternatively, if $\beta\leq 0$, then the smooth pasting fails and there is no solution to (\ref{HJB1}) and associated conditions. In this case, $\thre =\infty$ and the option to repay the loan is never exercised. 
\label{priceop}
\end{prop}

Let us assume from now on that $S$ is the discounted price of the market portfolio, so that the equilibrium rate of return 
$\overline\mu_2$ on the asset $V$ satisfies the CAPM condition 
\begin{equation}
\label{CAPM}
\frac{\overline\mu_2-r}{\sigma_2}=\rho\frac{\mu_1-r}{\sigma_1}.
\end{equation}
The dividend rate paid by $V$ is then $\delta=\overline\mu_2-\mu_2$, and we have that 
\begin{equation}
\label{beta}
\beta=1-\frac{2}{\sigma_2}\left(\frac{\mu_2-r}{\sigma_2}-\rho\frac{\mu_1-r}{\sigma_1}\right)=1+\frac{2\delta}{\sigma^2_2}>0.
\end{equation}

Because the bank is well--diversified and can hedge in the financial market by directly trading the asset $V$, the cost $C_{t_{0}}=C(V_{t_{0}})$ of granting the repayment option is given by the complete market price of a perpetual barier--type call option on $e^{r(t-t_0)}V_t$ with strike $e^{\alpha(t-t_0)}L$ exercised at the borrower's optimal exercise boundary obtained in Proposition \ref{priceop}. In other words, denoting by $Q$ the unique risk--neutral measure for the complete market consisting of $S$ and $V$, we have the following result.
\begin{prop} Assuming that the borrower exercises the repayment option optimally according to Proposition \ref{priceop}, the cost of this option for the bank is given by
\begin{equation}
\label{cost}
C(v)=\left\{\begin{aligned}
& (V^\ast-L)\mathbb E^Q\left[\mathbf{1}_{\{\tau^{*}<\infty\}}\right], \quad\mbox{if}\quad v<V^\ast  \\
& v-L, \qquad\qquad\qquad\qquad \mbox{if}\quad v \geq V^\ast
\end{aligned}\right.
=\left\{\begin{aligned}
& (V^\ast-L) \left(\frac{v}{V^{\ast}}\right)^{\beta}\quad\mbox{if}\quad v<V^\ast  \\
& v-L, \quad\qquad\qquad\mbox{if}\quad v \geq V^\ast.
\end{aligned}\right.
\end{equation}
\end{prop}

\begin{proof} Observe that the risk--neutral dynamics for $V$ is 
\begin{equation}
dV_t=-\delta V_{t}dt + \sigma_{2}V_{t}dW_{t}^{Q},
\end{equation}
where $W^{Q}$ is a Brownian motion. Therefore $E^Q\left[\mathbf{1}_{\{\tau^{*}<\infty\}}\right]$ corresponds to the risk--neutral probability that the geometric Brownian motion $V_{t}$ started at $V_{t_{0}}=v$ will cross the barrier $V^{\ast}$ in a finite time. The result then follows from standard Laplace transform techniques. 
\end{proof}

We can now use \eqref{relation} and \eqref{cost} to establish that 
\begin{equation}
\label{fee}
c=L+C_{t_0}-V_{t_0}=\left\{\begin{aligned}
& L+(V^\ast-L)\left(\frac{V_{t_{0}}}{V^{\ast}}\right)^{\beta}-V_{t_0}, \quad\quad\mbox{if}\quad V_{t_0}<V^\ast \\
& 0, \quad\qquad\qquad\qquad\qquad\quad\quad\quad\quad\mbox{if}\quad V_{t_0} \geq V^\ast.
\end{aligned}\right.
\end{equation}

As shown in Proposition 3.5 of \cite{Henderson07}, one can find by direct differentiation of expressions \eqref{shold} and \eqref{indiff} that both the threshold 
$V^{*}$ and the indifference value $p(v)$ are increasing in $\rho^{2}$. In other words, all things being equal,  a higher degree of market incompleteness, expressed as a smaller absolute value for the correlation between the stock and the market portfolio, lead the client to exercise the option to repay the loan earlier than in the complete market case, resulting in a smaller indifference value for the option. Similarly both $V^{*}$ and $p(v)$ are decreasing in $\gamma$, meaning that a higher degree of risk aversion has a similar effect in decreasing the value of the option to repay the loan. These properties carry over to the loan fee obtained in expression \eqref{fee}, as we establish in the next proposition. 

\begin{prop}
\label{dep-inf}
The loan fee:
\begin{enumerate}
\item decreases as the risk aversion $\gamma$ increases;
\item decreases as the dividend rate $\delta$ increases;
\item increases as $\rho^{2}$ increases.
\end{enumerate}
Moreover, its limiting values either as $\rho^{2}\to1$ or $\gamma\to 0$ coincide and are given by
\begin{equation}
\label{fee_limit}
c=\left\{\begin{aligned}
& L+(\widetilde V-L)\left(\frac{V_{t_{0}}}{\widetilde V}\right)^{\beta}-V_{t_0}, \quad\quad\mbox{if}\quad V_{t_0}<V^\ast \\
& 0, \quad\qquad\qquad\qquad\qquad\quad\quad\quad\quad\mbox{if}\quad V_{t_0} \geq V^\ast.
\end{aligned}\right.
\end{equation}
where $\widetilde V=\frac{\beta}{\beta-1}L=\left(1+\frac{\sigma^{2}_{2}}{2\delta}\right)L$. 
\end{prop}
\begin{proof}
The first part of the proposition follows by explicit differentiation of expressions \eqref{fee}, \eqref{shold} and \eqref{beta}. For the second part, observe that it follows from our equilibrium condition \eqref{CAPM} that both limiting thresholds in Proposition 3.5 of \cite{Henderson07} are given by $\widetilde V=\frac{\beta}{\beta-1}L$. Substituting expression \eqref{beta} for $\beta$ then completes the proof. 
\end{proof}
We conclude this section by observing that the limiting threshold $\widetilde V$ corresponds to the complete market threshold $a_{0}$ found in \cite{XiaZhou07} using risk-neutral valuation arguments. Consequently, provided $\alpha=r$, the complete market, risk--neutral setting can be recovered as a special case of our results. 

\section{Finite maturity}

\subsection{The free boundary problem}
\label{finite}

Consider $T<\infty$ and define the value function (see \cite{Merton69})
\begin{equation}
\label{merton}
 M(t,x)=\sup_{\pi\in {\cal A}_{[t,T]}}\mathbb{E}[-e^{-\gamma X^{\pi}_{T}}|X^\pi_{t}=x]=-e^{-\gamma x}e^{-\frac{(\mu_1-r)^2}{2\sigma^2}(T-t)},
\end{equation}
for $t_0\leq t\leq T$, where $X^\pi_t$ follows the dynamics \eqref{wealth} and ${\cal A}_{[t,T]}$ is the set of admissible investment policies on the interval $[t,T]$, which we take to be progressively measurable processes satisfying the integrability condition 
\[\mathbb{E}\left[\int_t^T \pi_s^2ds\right]<\infty.\]
As before, let the repayment time by a stopping time $\tau$ and assume that the borrower will add the discounted payoff 
$(V_\tau-e^{(\alpha-r)(\tau-t_0)}L)^+$ to his discounted wealth $X^\pi_\tau$ at time $\tau$ and then invest optimally until time $T$. Accordingly, having taken the stock loan at time $t_0$, the borrower consider the following optimization problem:
\begin{equation}
u(t_0,x,v)=\sup_{\tau\in{\cal T}[t_{0},T]}\sup_{\pi\in {\cal A}_{[t,\tau]}}\mathbb{E}[M(\tau,X^{\pi}_{\tau}+(V_\tau-e^{(\alpha-r)(\tau-t_0)}L)^+)|X^\pi_{t_0}=x,V_{t_0}=v],
\label{optimal}
\end{equation}
where ${\cal T}[t_{0},T]$ denotes the set of stopping times in the interval $[t_{0},T]$. The indifference value for the repayment option is then given by the amount  
$p$ satisfying  
\begin{equation}
\label{indiff_fin}
M(t_0,x)=u(t_0,x-p,v).
\end{equation}

It follows from the dynamic programming principle that the value function $u$ solves the free boundary problem
\begin{equation}
\label{comp1}
\left\{\begin{aligned}
 \frac{\partial u}{\partial t}+\displaystyle{\sup_\pi {\cal L}^\pi u} & \leq 0, \\
 u(t,x,v) & \geq \Lambda(t,x,v), \\
 \left(\frac{\partial u}{\partial t}+\displaystyle{\sup_\pi {\cal L}^\pi u}\right)&\cdot\left(u-\Lambda\right) =0,
\end{aligned}\right.
\end{equation}
for $(t,x,v)\in [t_0,T)\times\mathbb{R}\times (0,\infty)$, where
\[{\cal L}^\pi=(\mu_2-r)v\frac{\partial }{\partial v}+\frac{\sigma_2^2v^2}{2}\frac{\partial^2 }{\partial v^2}+
\pi(\mu_1-r)\frac{\partial }{\partial x}+\rho\pi\sigma_1\sigma_2 v\frac{\partial^2 }{\partial x\partial v}+\frac{\pi^2\sigma_1^2}{2}
\frac{\partial^2 }{\partial x^2}\]
is the infinitesimal generator of $(X^\pi,V)$ and 
\[\Lambda(t,x,v)=M(t,x+(v-e^{(\alpha-r)(t-t_0)}L)^+)\]
is the utility obtained from exercising the repayment option at time $t$. The boundary conditions for Problem \eqref{comp1} are 
\begin{equation}
\begin{split}
u(T,x,v)&=-e^{-\gamma[x+(v-e^{(\alpha-r)(T-t_0)}L)^+]} \\
u(t,x,0)&=-e^{-\gamma x}e^{-\frac{(\mu_1-r)^2}{2\sigma^2}(T-t)}.
\end{split}
\end{equation}

Using the factorization
\begin{equation}
\label{factor}
u(t,x,v)=M(t,x)F(t,v)^{\frac{1}{1-\rho^2}},
\end{equation}
we find that the corresponding free boundary problem for $F$ becomes 
\begin{equation}
\label{comp2}
\left\{\begin{aligned}
\frac{\partial F}{\partial t}+{\cal L}^0F&\geq 0, \\
 F(t,v)&\leq \kappa(t,v), \\
\left(\frac{\partial F}{\partial t}+{\cal L}^0F\right)&\cdot(F-\kappa)=0,
\end{aligned}\right.
\end{equation}
for $(t,v)\in[t_0,T)\times (0,\infty)$, where
\begin{equation}
\label{L0}
{\cal L}^0=\left[\mu_2-r-\rho\frac{\mu_1-r}{\sigma_1}\sigma_2\right]v\frac{\partial }{\partial v}+\frac{\sigma_2^2v^2}{2}
\frac{\partial^2 }{\partial v^2}
\end{equation}
and 
\begin{equation}
\kappa(t,v)=e^{-\gamma(1-\rho^2)(v-e^{(\alpha-r)(t-t_0)}L)^+}.
\end{equation}
The boundary conditions for Problem \eqref{comp2} are
\begin{equation}
\begin{split}
F(T,v) &=e^{-\gamma (1-\rho^2)(v-e^{(\alpha-r)(T-t_0)}L)^+} \\
F(t,0) &= 1.
\end{split}
\end{equation}

Since  Problem \eqref{comp2} is independent of $X$ and $S$, we define the borrower's optimal exercise boundary as the function
\begin{equation}
\label{Vstar}
V^\ast(t)=\inf\left\{v\geq 0: F(t,v)= \kappa(t,v)\right\}
\end{equation}
and the optimal repayment time as
\begin{equation}
\label{tau2}
\tau^\ast=\inf\left\{t_0\leq t\leq T: V_t = V^\star(t)\right\}.
\end{equation}

It follows from the definition \eqref{indiff_fin} and the factorization \eqref{factor} that the indifference value for the repayment option is given by $p=p(t_0,V_{t_0})$ where
\begin{equation}
p(t,v)=-\frac{1}{\gamma(1-\rho^2)}\log F(t,v).
\end{equation}
Therefore, the original free boundary problem can be rewritten in terms of the indifference value as 
\begin{equation}
\label{comp3}
\left\{\begin{aligned}
\frac{\partial p}{\partial t}+{\cal L}^0p-\frac{1}{2}\gamma(1-\rho^2)\sigma^2_2v^2\left(\frac{\partial p}{\partial v}\right)^2&\leq 0, \\
 p(t,v)&\geq \left(v-e^{(\alpha-r)(t-t_0)}L\right)^+, \\
\left[\frac{\partial p}{\partial t}+{\cal L}^0p-\frac{1}{2}\gamma(1-\rho^2)\sigma^2_2v^2\left(\frac{\partial p}{\partial v}\right)^2\right]&\cdot(p-(v-e^{(\alpha-r)(t-t_0)}L)^+)=0,
\end{aligned}\right.
\end{equation}
Similarly, the optimal exercise time $\tau^\ast$ can be expressed in terms of $p$ as follows:
\begin{equation}
\label{tau3}
\tau^\ast=\inf\left\{t_0\leq t \leq T: p(t,V_t)=(V_t-e^{(\alpha-r)(t-t_0)}L)^+\right\}
\end{equation}

Once we find the optimal exercise boundary $V^\ast(t)$, say by solving problem \eqref{comp2} numerically, we can calculate the bank's cost of granting the repayment option as the risk--neutral value of a barier--type call option on 
$e^{r(t-t_0)}V_t$ with strike $e^{\alpha(t-t_0)}L$ and maturity $T$ exercised at the barrier $e^{r(t-t_0)}V^\ast(t)$. In other words 
\begin{align}
C_{t_{0}}=C(t_{0},v)&=E^Q\left[\left.e^{-r(\tau-t_{0})}\left(e^{r(\tau-t_{0})}V^\ast(t)-e^{\alpha(\tau-t_{0})}L\right)^{+}\mathbf{1}_{\{\tau^{*}<\infty\}}\right|V_{t_{0}}=v\right]\\
&=E^Q\left[\left.e^{-\widehat r(\tau-t_{0})}\left(e^{(r-\alpha)(\tau-t_{0})}V^\ast(t)-L\right)^{+}\mathbf{1}_{\{\tau^{*}<\infty\}}\right|V_{t_{0}}=v\right]\\
&=E^Q\left[\left.e^{-\widehat r(\tau-t_{0})}\left(\widehat V^\ast(t)-L\right)^{+}\mathbf{1}_{\{\tau^{*}<\infty\}}\right|V_{t_{0}}=v\right]
\end{align}
where $\widehat r=r-\alpha$ and $\widehat V(t)=e^{\widehat r(\tau-t_{0})}V^\ast(t)$. Denoting $\widehat V_{t}=e^{(r-\alpha)(\tau-t_{0})}V_{t}$, it is trivial to see that
$\tau^{\ast}$ defined in \eqref{tau2} can be written as
\begin{equation}
\tau^\ast=\inf\left\{t_0\leq t\leq T: V_t = V^\star(t)\right\}=\inf\left\{t_0\leq t\leq T: \widehat V_t = \widehat V^\star(t)\right\}
\end{equation}
Therefore, since the risk--neutral dynamics for the process $\widehat V_{t}$ is 
\begin{equation}
d\widehat V_t=(\widehat r-\delta) V_{t}dt + \sigma_{2}V_{t}dW_{t}^{Q},
\end{equation}
we have that the function $C(t,v)$ satisfies the Black--Scholes PDE
\begin{equation}
\label{pde}
\frac{\partial C}{\partial t}+(r-\alpha-\delta)v\frac{\partial C}{\partial v}+\frac{\sigma^2_2v^2}{2}\frac{\partial^2 C}{\partial v^2}=(r-\alpha)C
\end{equation}
over the domain ${\cal D}=\{(t,v): t_0\leq t \leq T, 0\leq v\leq V^\ast(t)\}$, subject to the boundary conditions
\begin{equation}
\label{boundary}
\begin{split}
C(t,0)&=0, \quad\qquad\qquad\qquad t_0\leq t\leq T, \\
C(t,\widehat V^\ast(t)) &= (\widehat V^\ast(t)-L)^+, \qquad t_0\leq t\leq T,\\
C(T,v) &= (v-L)^+, \qquad \qquad 0\leq v\leq \widehat V^*(T)
\end{split}
\end{equation}

As before, once we calculate the cost $C_{t_0}$, the fee to be charged for the loan is given by \eqref{relation}. It is easy to see that the cost $C_{t_0}$, and consequently the fee $c$, increase if the optimal exercise boundary $V^*(t)$ is shifted upward and decrease otherwise. 

\subsection{Properties of the loan fee}

In this section we investigate how the loan fee to be charged by the bank depends on the underlying parameters. We will always assume that the interest rate $r$, the expected return $\mu_1$ and volatility $\sigma_1$ for the market portfolio $S$, the loan interest rate $\alpha$, and loan amount $L$ are fixed. 
On the other hand, we treat the risk aversion $\gamma$, the dividend rate $\delta$, the correlation $\rho$, and the underlying asset volatility $\sigma_2$ as variable parameters. We then perform comparative statics, that is, we change each of these parameters while keeping the others constant and analyze the corresponding behavior of the loan fee.  

Observe that for each choice of values for $\delta$, $\sigma_2$ and 
$\rho$ the expected return $\mu_2$ is automatically determined by the assumption that asset prices are in equilibrium. For simplicity, we continue to assume that $S$ is the discounted price of the market portfolio, so that the CAPM condition \eqref{CAPM} holds and we have that 
\begin{equation}
\label{mu2}
\mu_2=\rho\frac{\mu_1-r}{\sigma_1}\sigma_2+r-\delta.
\end{equation}

The behavior of the loan fee with respect to the underlying parameters is established
in the next proposition, which we prove using the same technique as in \cite{LeungSircar09}, but adapted the problem at hand.

\begin{prop}
\label{proposition}
The loan fee $c$:
\begin{enumerate}
\item decreases as the risk aversion $\gamma$ increases;
\item decreases as the dividend rate $\delta$ increases;
\item increases as $\rho^{2}$ increases;
\end{enumerate}
\end{prop}
\begin{proof}
Observe first that for fixed values of $L,\alpha$ and $r$, it follows from \eqref{tau3} that a smaller indifference value leads to a smaller optimal exercise time, which in turns implies a lower optimal exercise boundary and consequently a smaller loan fee.  To establish how the indifference value changes with the underlying parameters, we use the comparison principle for the variational inequality 
\begin{equation}
\label{var}
\min\left\{-\frac{\partial p}{\partial t}-{\cal L}^0p+\frac{1}{2}\gamma(1-\rho^2)\sigma^2_2v^2\left(\frac{\partial p}{\partial v}\right)^2,p(t,v)- \left(v-e^{(\alpha-r)(t-t_0)}L\right)^+\right\}=0
\end{equation}
which is known to be equivalent to \eqref{comp3}.

For item (1), observe that the variational inequality \eqref{var} depends on $\gamma$ only through the nonlinear term 
\begin{equation}
\label{nonlinear}
\frac{1}{2}\gamma(1-\rho^2)\sigma^2_2v^2\left(\frac{\partial p}{\partial v}\right)^2.
\end{equation}
Since this is increasing in $\gamma$, it follows that $p$ is decreasing in $\gamma$. 

For item (2), observe first that $\frac{\partial p}{\partial v}\geq 0$, because $u(,t,x,v)$ defined in \eqref{optimal} (and consequently $p(t,v)$) is an increasing function $v$. Next, recalling the definition of ${\cal L}^{0}$ in \eqref{L0}, we see that the variational inequality \eqref{var} depends on $\delta$ through the term
\[-\left[\mu_2-r-\rho\frac{\mu_1-r}{\sigma_1}\sigma_2\right]\frac{\partial p}{\partial v}=\delta \frac{\partial p}{\partial v},\]
on account of \eqref{mu2}. Since this is increasing is $\delta$, we have that $p$ is decreasing in $\delta$.

Similarly for item (3), using \eqref{L0} we see that the variational inequality \eqref{var} depends on $\rho$ through the term 
\[-\left[\mu_2-r-\rho\frac{\mu_1-r}{\sigma_1}\sigma_2\right]\frac{\partial p}{\partial v}+\frac{1}{2}\gamma(1-\rho^2)\sigma^2_2v^2\left(\frac{\partial p}{\partial v}\right)^2.\]
By virtue of \eqref{mu2}, we then see that the dependence on $\rho$ reduces to the nonlinear term \eqref{nonlinear}. Therefore the indifference price is a symmetric function of $\rho$, and increases as $\rho^2$ increases from $0$ to $1$.
\end{proof}

Notice that the variational inequality \eqref{var} depends on $\sigma_2$ through the term 
\begin{equation}
-\frac{\sigma_2^2v^2}{2}\frac{\partial^2p }{\partial v^2}+\frac{1}{2}\gamma(1-\rho^2)\sigma^2_2v^2\left(\frac{\partial p}{\partial v}\right)^2.
\end{equation}
Since this is not necessarily monotone in $\sigma_2$, we cannot expect the indifference value, and consequently the loan fee $c$, to be monotone function of the underlying stock volatility.

Regarding the behavior of the loan fee with respect to the maturity length of the loan, one intuitively expects that a longer maturity increases the optionality of the repayment and should contribute to higher fee. As establish in the next proposition, this is indeed the case provided we can ignore the effects of interest rates.  

\begin{prop} If $\alpha=r$, the loan fee is an increasing function of the maturity $T$. 
\label{time}
\end{prop}
\begin{proof} The solution to problem \eqref{comp2} admits a probabilistic representation (see \cite{ObermanZari03}) of the form 
\[F(t,v)=\inf_{\tau\in{\cal T}[t,T]} E^0[\kappa(\tau,V_\tau)|V_t=v],\]
where $E^0[\cdot]$ denotes the expectation operator under the {\em minimal martingale measure} $Q^0$ defined by 
\begin{equation}
\frac{dQ^0}{dP}=e^{-\frac{\mu_1-r}{\sigma_1}W_T-\frac{1}{2}\frac{(\mu_1-r)^2}{\sigma_1^2}T}.
\end{equation}   
When $\alpha=r$, we can use the fact that $V_t$ is a time--homogeneous diffusion to obtain that
\begin{align*}
F(t,v)&=\inf_{\tau\in{\cal T}[t,T]} E^0[e^{-\gamma(1-\rho^2)(V_\tau-L)^+}|V_t=v] \\
&= \inf_{\tau\in{\cal T}_{t_0,T-t+t_0}} E^0[e^{-\gamma(1-\rho^2)(V_\tau-L)^+}|V_{t_0}=v].
\end{align*}
For any $s\leq t$ we have that ${\cal T}[t_0,T-t+t_0]\subset {\cal T}[t_0,T-s+t_0]$, so $F(s,v)\leq F(t,v)$. Now fix $v>0$ and suppose that it is optimal to exercise at $(s,v)$, that is, $F(s,v)=k(s,v)$. Using the fact that $F$ is increasing in time (as we just established), we have that  
\[e^{-\gamma(1-\rho^2)(v-L)^+}=k(s,v)=F(s,v)\leq F(t,v)\leq k(t,v)=e^{-\gamma(1-\rho^2)(v-L)^+},\]
so that $F(t,v)=k(t,v)$, which implies that it is also optimal to exercise at $(t,v)$. This means that, for each fixed $T$, the optimal exercise boundary $V^{*}(t)$ is a decreasing function of time, and consequently an increasing function of the time--to--maturity parameter $(T-t)$. Therefore, as we modify the problem by increasing the maturity $T$, the optimal exercise boundary shits upwards, leading to a higher cost for the bank and a higher loan fee.  
\end{proof}

\section{Numerical Results}

\subsection{Infinite Maturity}

The only numerical step involved in this case consists of finding the value of the threshold $\thre$ by solving the nonlinear equation \eqref{shold} for given parameter values. We can then use expression \eqref{fee} to find the loan fee $c$. For comparison, we calculate the corresponding loan fee in a complete market scenario using the formulas found in \cite{XiaZhou07}. Observe that we always need to use $r=\alpha$ as explained in Section \ref{infinite} in order to maintain time-homogeneity.

We start by calculating the value of the loan fee $c$ for a range of loan amounts $L$ and four different sets of model parameters. The results are summarized in Table \ref{tableInfMat} and correspond to the following cases:

\begin{enumerate}
\item Complete market with $\sigma_{2}=0.15, \delta=0, r=0.05,\alpha=0.05$ and $V_{t_{0}}=100$. This corresponds to case (a) of Theorem 3.1 in 
\cite{XiaZhou07}, that is, $\delta=0$ and $\alpha-r=0<\sigma^{2}_{2}/2$. Because the stock pays no dividend and the excess interest rate on the loan is small, it follows that the option to repay the loan has the same value as the stock itself (that is $C_{t_{0}}=V_{t_{0}}$), which leads to $c=L$. In other words, the bank has no incentive to provide the loan and charges a fee exactly equal to the loan amount. In effect, the client gives away the stock and receives a perpetual American option with an infinite exercise threshold. 
\item Incomplete market with $\sigma_{2}=0.15, \delta=0, r=0.05,\alpha=0.05, V_{t_{0}}=100,  \rho=0.9$ and $\gamma=0.01$. This is the incomplete market analogue of the previous case. We see that incompleteness  and risk aversion  lead to a finite exercise threshold $V^{*}$ even when the stock pays no dividend and the excess interest rate on the loan is small. In other words, the option to repay the loan is exercised sooner, and consequently has a smaller value for the client, than in the complete market case. As a consequence, the bank has a smaller cost for providing the loan and can charge a reduced fee $c<L$.
\item Complete market with $\sigma_{2}=0.15, \delta=0.05, r=0.05,\alpha=0.05$ and $V_{t_{0}}=100$. This corresponds to case (b) of Theorem 3.1 in 
\cite{XiaZhou07}, since $\delta>0$. To calculate $c$, we first find the exercise threshold $a_{0}$ as in page 314 of \cite{XiaZhou07} for each value of $L$. If $a_{0}\leq V_{t_{0}}=100$ (which happens for low enough $L$) then $c=0$, meaning that the client receives $L$ in exchange of $V_{t_{0}}$ at no cost, and then immediately exercises the option to repay. In other words, there is no incentive for the client to seek the loan. On the other hand, if $a_{0}>V_{t_{0}}$, we calculate the fee $c$ using the formula at the end of page 316 of \cite{XiaZhou07}.
\item Incomplete market with  $\sigma_{2}=0.15, \delta=0.05, r=0.05,\alpha=0.05, V_{t_{0}}=100, \rho=0.9$ and $\gamma=0.01$. This is the incomplete market analogue of the previous case. As expected, we have that $V^{*}<a_{0}$, leading to a smaller fee $c$ charged by the bank.  
\end{enumerate}

\begin{table}
\caption{Loan fee  $c$ as for different loan amounts $L$ (infinite maturity)}
{\begin{tabular}{|c|c|c|c|c|c|c|c|c|c|}\hline
\multicolumn{2}{|c|}{$L$}&  50 & 60 &    70   &   80    & 90   & 100 & 110 & 120 \\ 
\hline
\multirow{1}{*}{Case 1} & c  & 50 & 60 &    70   &   80    & 90   & 100 & 110 & 120  \\ \hline
\multirow{2}{*}{Case 2} & c  & 31.0528 & 39.5086 &  48.1242 &  56.8653 & 65.7084 &  74.6363 & 83.6361 & 92.6978 \\ 
& $V^{*}$  & 263.8914 & 292.8058 &319.9876   &345.8010  &370.4988  &394.2648   &417.2377  &439.5251 \\ 
\hline
\multirow{2}{*}{Case 3}& $c$  & 0.0000  & 0.0000 & 0.0000  &0.0000   &1.9041  &7.4530   &14.8794  & 23.3145 \\ 
& $a_{0}$  & 61.2500  &73.5000  &85.7500   &98.0000   &110.2500  &122.5000   & 134.7500 & 147.0000  \\ 
\hline
\multirow{2}{*}{Case 4}& c  & 0.0000 & 0.0000 & 0.0000 & 0.0000 & 1.9015 & 7.4510   & 14.8778  &23.3132\\
& $V^{*}$  & 61.1055 & 73.2926 &85.4688  & 97.6341 & 109.7885  & 121.9323  &134.0656  & 146.1884  \\
\hline
\end{tabular}}
\label{tableInfMat}
\end{table}

Next in Figure \ref{inFiMatFig} we illustrate the dependence of the loan fee upon the other model parameters. In particular, each curve on the top left plot represents $c$ as a function of $\gamma$ for a particular value of the correlation $\rho$. Similarly, each curve on the top right plot the loan fee value $c$ as a function of $\rho$ for a particular value of the risk aversion $\gamma$. Finally the curve on the bottom plot represents the loan fee value $c$ as a function of the dividend rate $\delta$. When not explicitly shown in the figure, the values of the remaining parameters are $\sigma_{2}$=0.15, $\delta$=0.05, $r=\alpha=0.05$, $L$=90, $V_{0}=100$, $\rho=0.9$ and $\gamma=0.01$.

\begin{figure}
\begin{center} 
\subfloat{\includegraphics[scale=0.5]{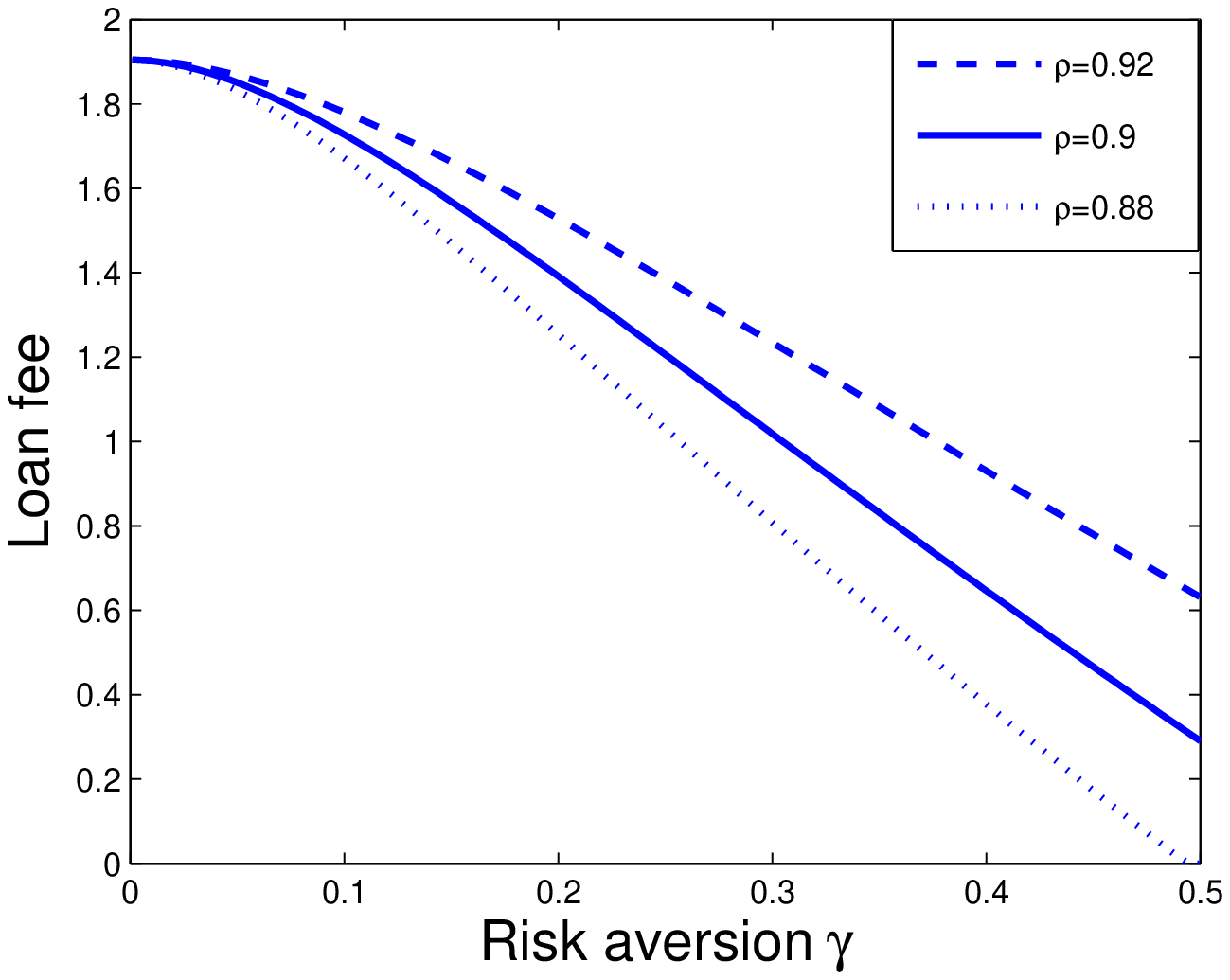}}
\subfloat{\includegraphics[scale=0.5]{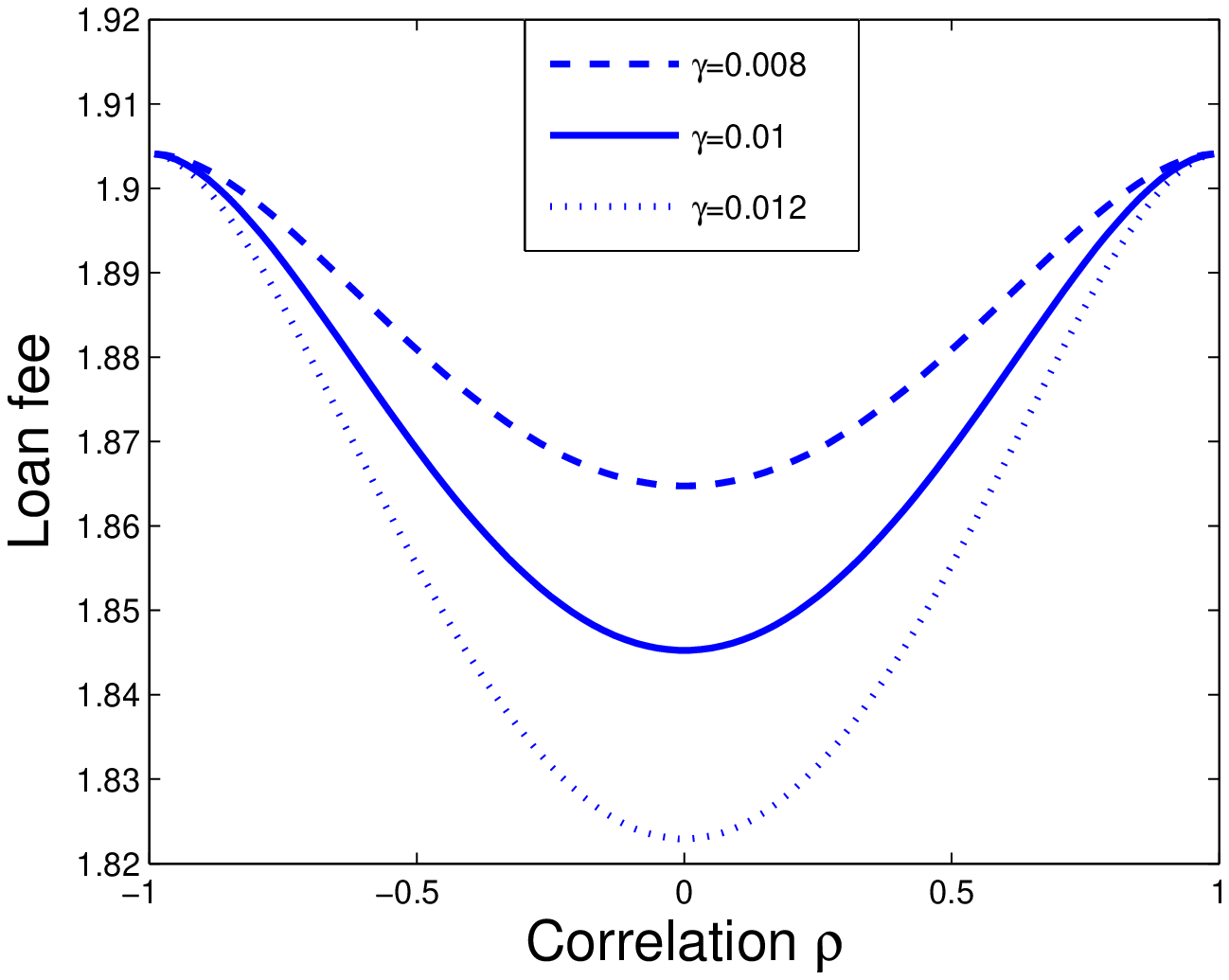}}\\
\subfloat{\includegraphics[scale=0.5]{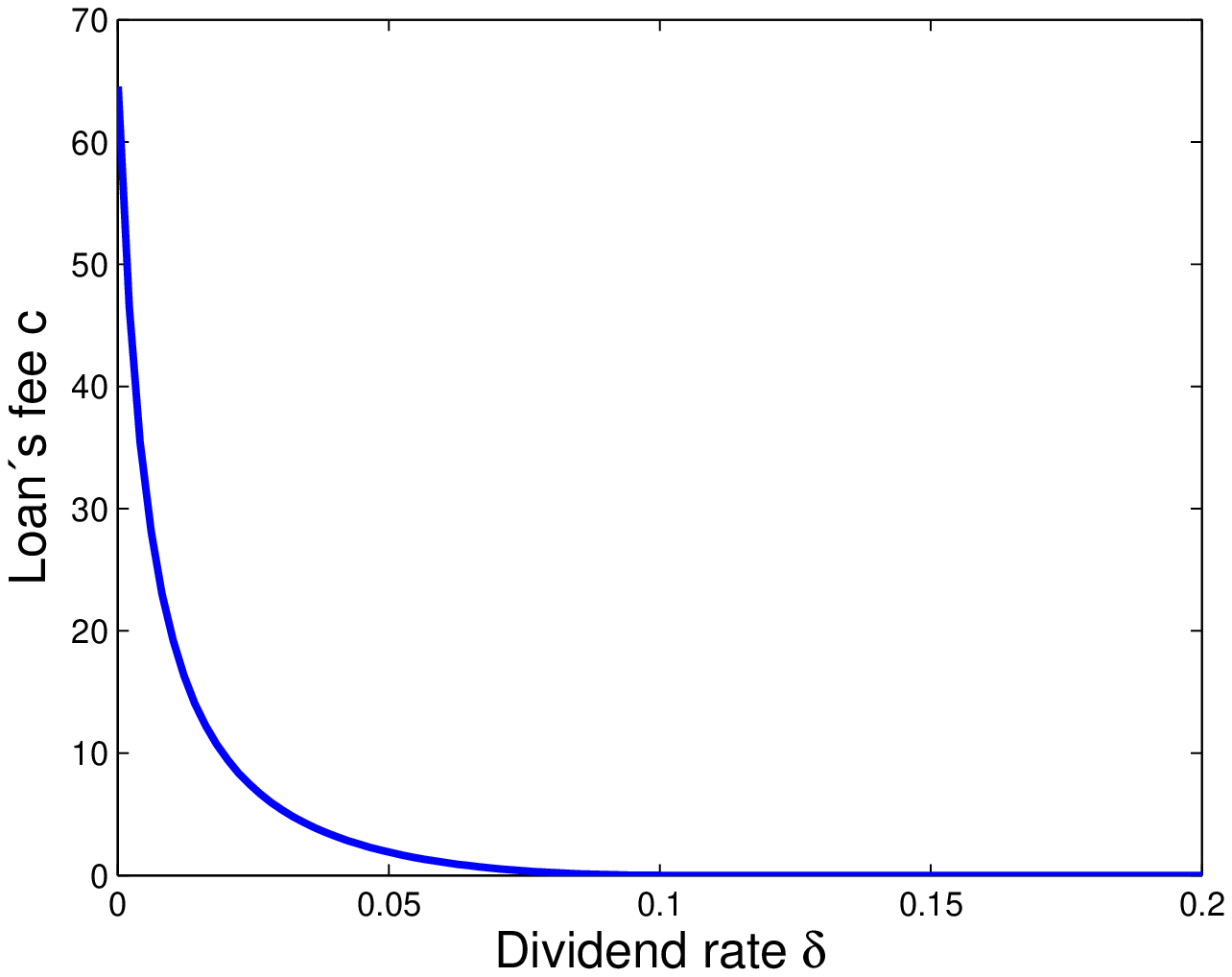}}
\caption{Dependence on model parameters for infinite maturity}
\label{inFiMatFig}
\end{center}
\end{figure}

Observe that the figures confirm the dependences established in Proposition \ref{dep-inf}, namely that the loan fee $c$ is decreasing in $\gamma$ and $\delta$ and increasing in $\rho^{2}$. Moreover, the limits as $\gamma\to 0$ and $\rho\to\pm 1$ coincide with the complete market, risk neutral value for the fee obtained in Case 3 of Table \ref{tableInfMat} for a loan amount $L=90$, namely $c=1.9041$. 

Observe further that in the incomplete market case, the fee increases sharply as $\delta\to 0$, but converges to the value $c=65.7048<90$ as obtained in Case 2 of Table \ref{tableInfMat} for $L=90$.

\subsection{Finite maturity}

As mentioned in Section \ref{finite}, the numerical procedures in this case are slightly more involved. First we use finite differences with projected successive--over--relaxation (PSOR) to solve the linear free boundary problem \eqref{comp2}. This yields a threshold function $\thre (t)$, which we then use to solve equation \eqref{pde} subject to the boundary conditions \eqref{boundary}, again by finite differences. 

To start with, Table \ref{tableFinMat} shows the loan fee $c$ for different loan amounts $L$, with the following parameter values: 
 $\sigma_{2}=0.4, \rho=0.4, \gamma=0.01, \delta=0.05, r=0.05, \alpha=0.07, V_{t_{0}}=100$ and $T=5$ (in years). Observe that we do not need to restrict ourselves to the case $r=\alpha$ as we did before, since the time--homogeneity property is not used in the finite--maturity case.
 
\begin{table}
\caption{Loan fee $c$ for different loan amounts $L$ (finite maturity)}
\begin{center}
{\begin{tabular}{|c|c|c|c|c|c|c|c|c|}\hline
$L$  & 50 & 60 &    70   &   80    & 90   & 100 & 110& 120\\ \hline
$c$   & 0.0000  &0.0000   &0.0000  &1.0667 & 4.1073 &9.3487 &16.0344   &23.8156 \\ \hline
\end{tabular}}
\end{center}
\label{tableFinMat}
\end{table}

Next in Figure \ref{finite_fig}, we illustrate in detail the dependence upon the model parameters analyzed in Propositions \ref{proposition} and \ref{time}. We use 
$T=5$, $L=80$, $\sigma_{2}=0.4$, $r=0.05$, $\alpha=0.07$, $\delta=0.05$, and $\rho=0.4$ unless otherwise specified. Each curve on the left side represents an optimal exercise boundary $\thre (t)$ for $V_0=100$ , whereas each curve on the right side represent the loan fee $c$ as function of $V_{0}$, for the particular set of parameter values described below:
\begin{enumerate}
\item For the top row, we use $\gamma=0.01,0.05,0.08$ and find that both the optimal exercise boundary and the loan fee decrease as risk aversion increases, in agreement with item 1 of Proposition \ref{proposition}.
\item For the second row, we use $\delta=0.05,0.1,0.15$ and find that both the optimal exercise boundary and the loan fee decrease as the dividend rate increases, in agreement with item 2 of Proposition \ref{proposition}.
\item For the third row, we use $\rho=0.05,0.4,0.9$ and find that both the optimal exercise boundary and the loan fee increase as correlation increases, in agreement with item 3 of Proposition \ref{proposition}.
\item For the bottom, row we use $\alpha=r=0.05$ and find that the optimal exercise boundary is strictly decreasing with respect to time-to-maturity $(T-t)$, in agreement with Proposition \ref{time}
\end{enumerate}

\begin{figure}
\begin{center}
\subfloat{\includegraphics[scale=0.45]{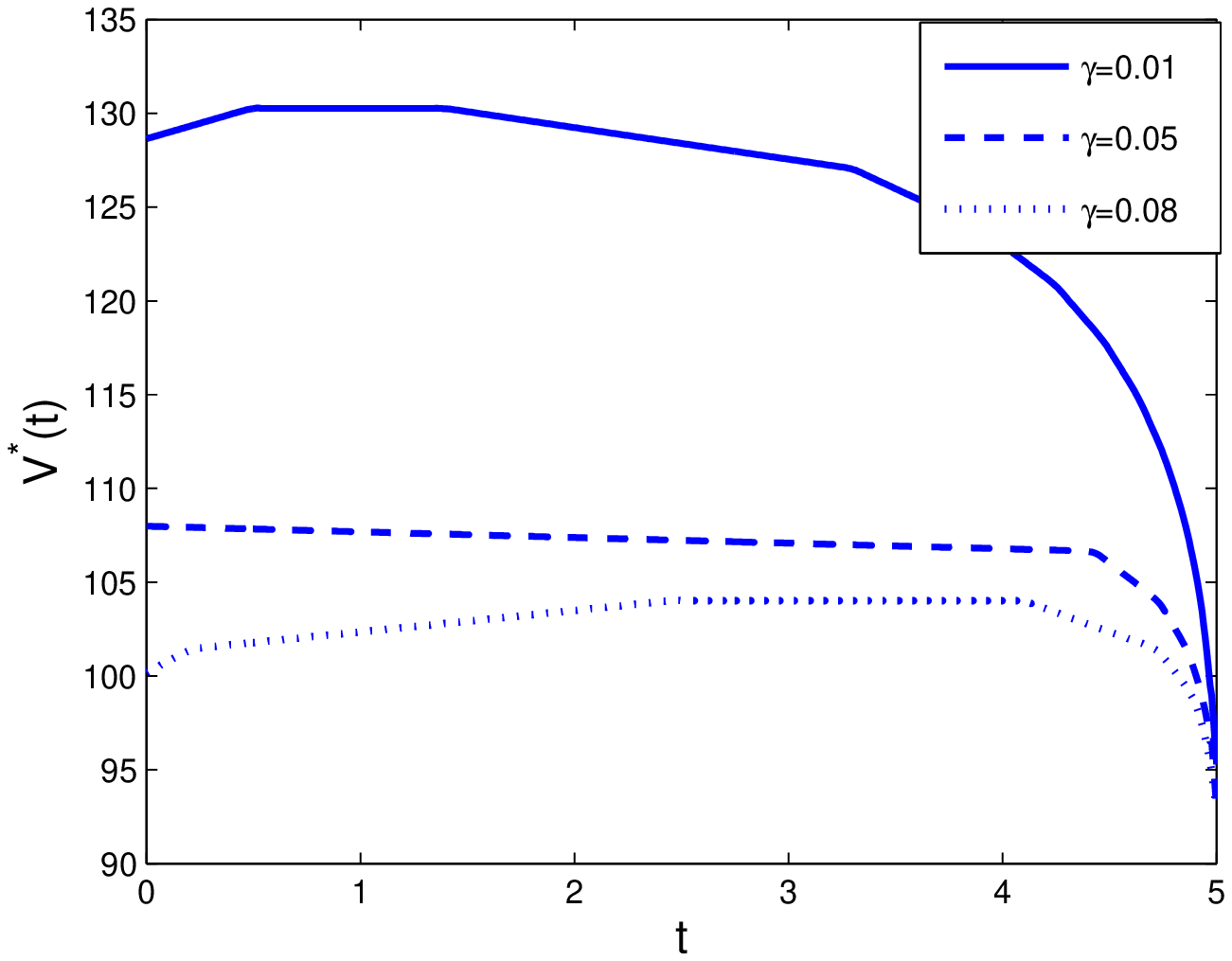}}
\subfloat{\includegraphics[scale=0.45]{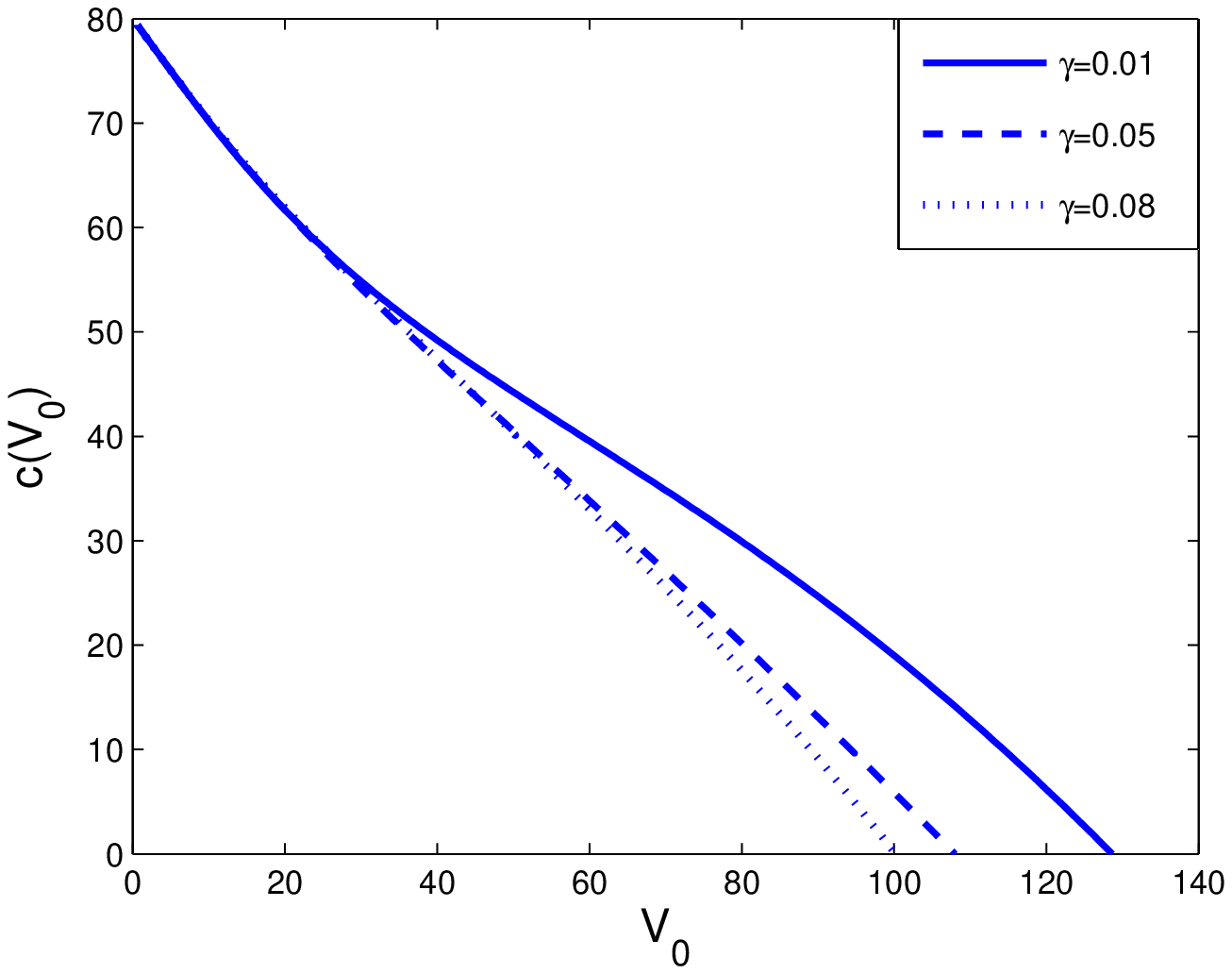}}\\
\subfloat{\includegraphics[scale=0.45]{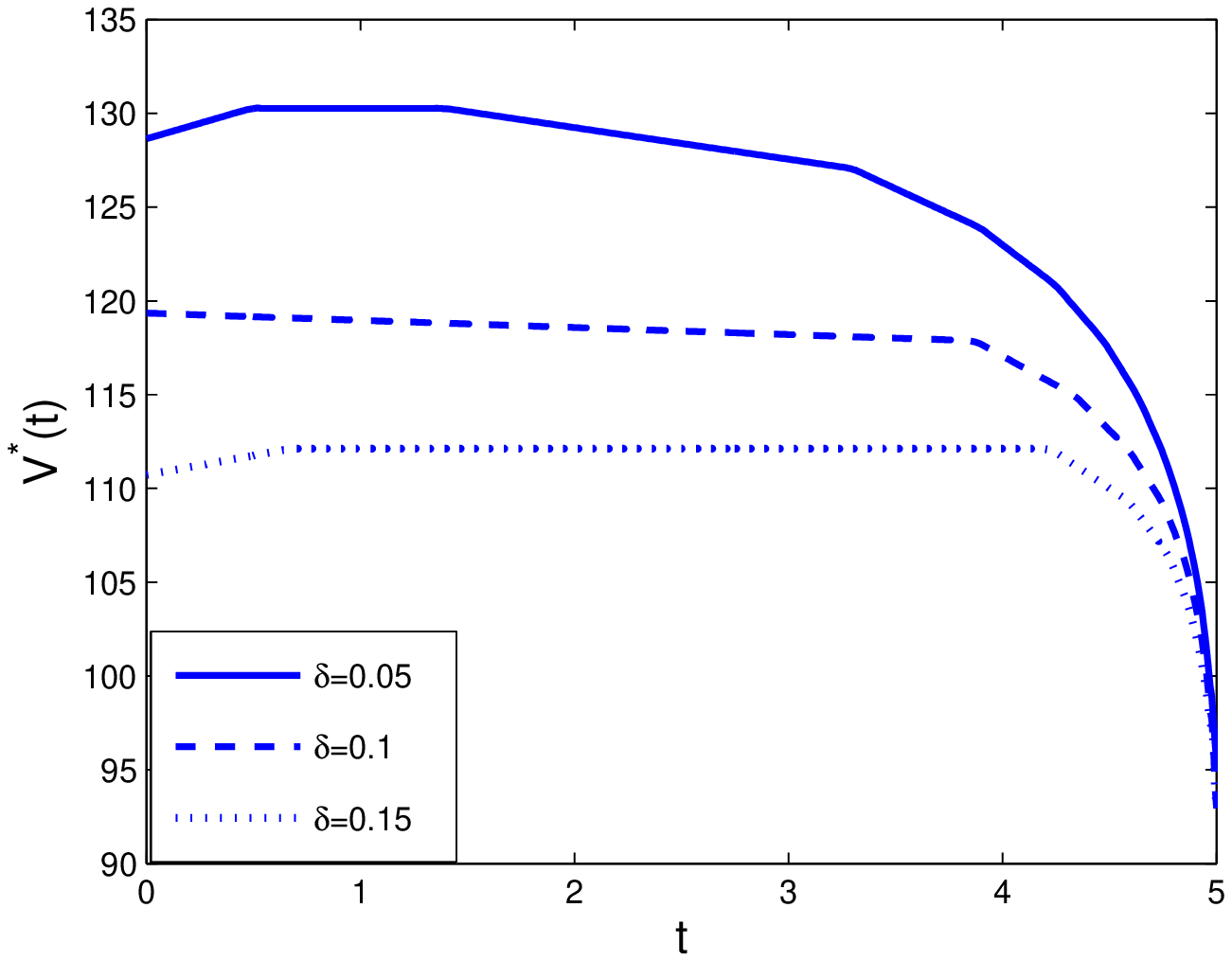}}
\subfloat{\includegraphics[scale=0.45]{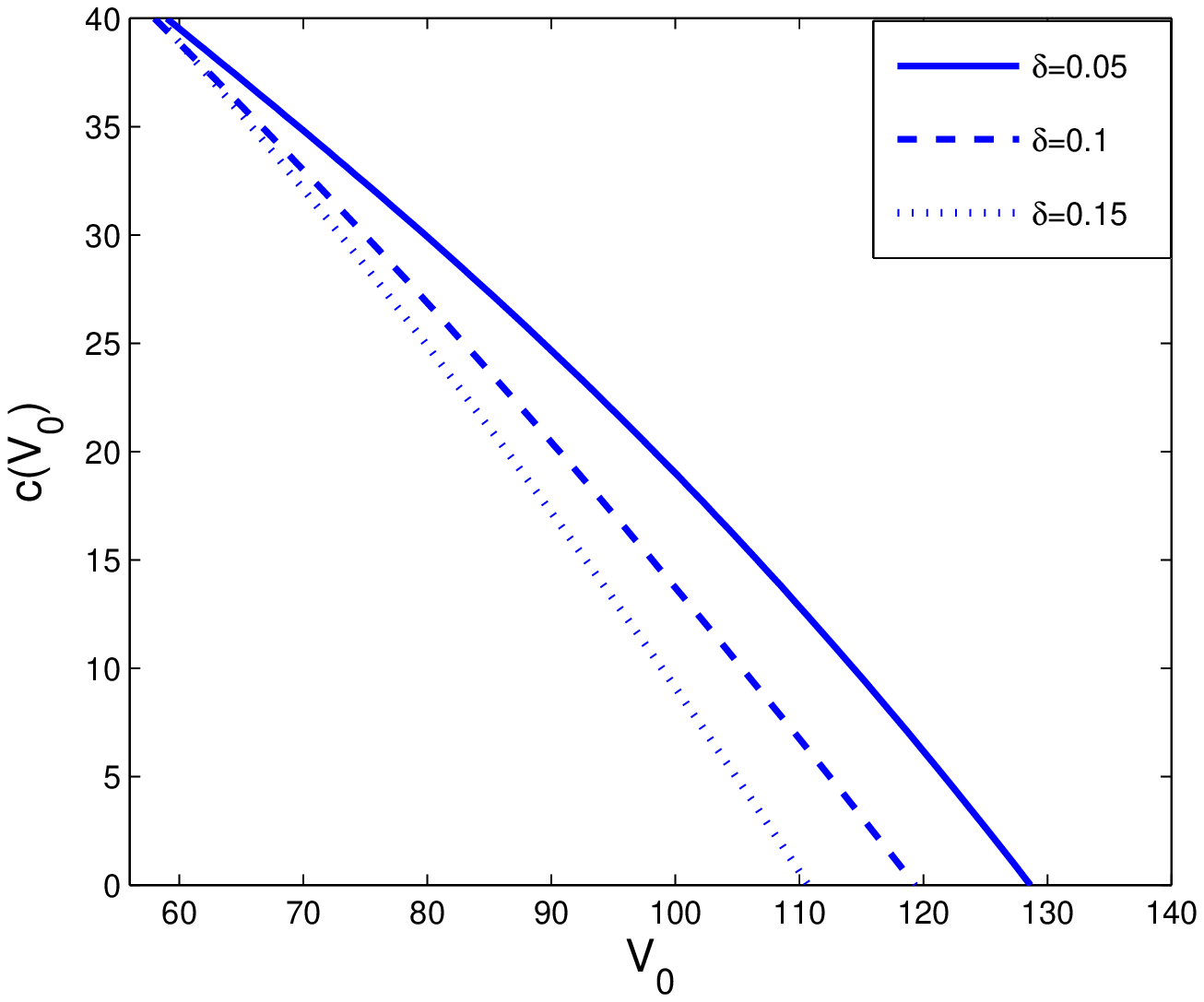}}\\
\subfloat{\includegraphics[scale=0.45]{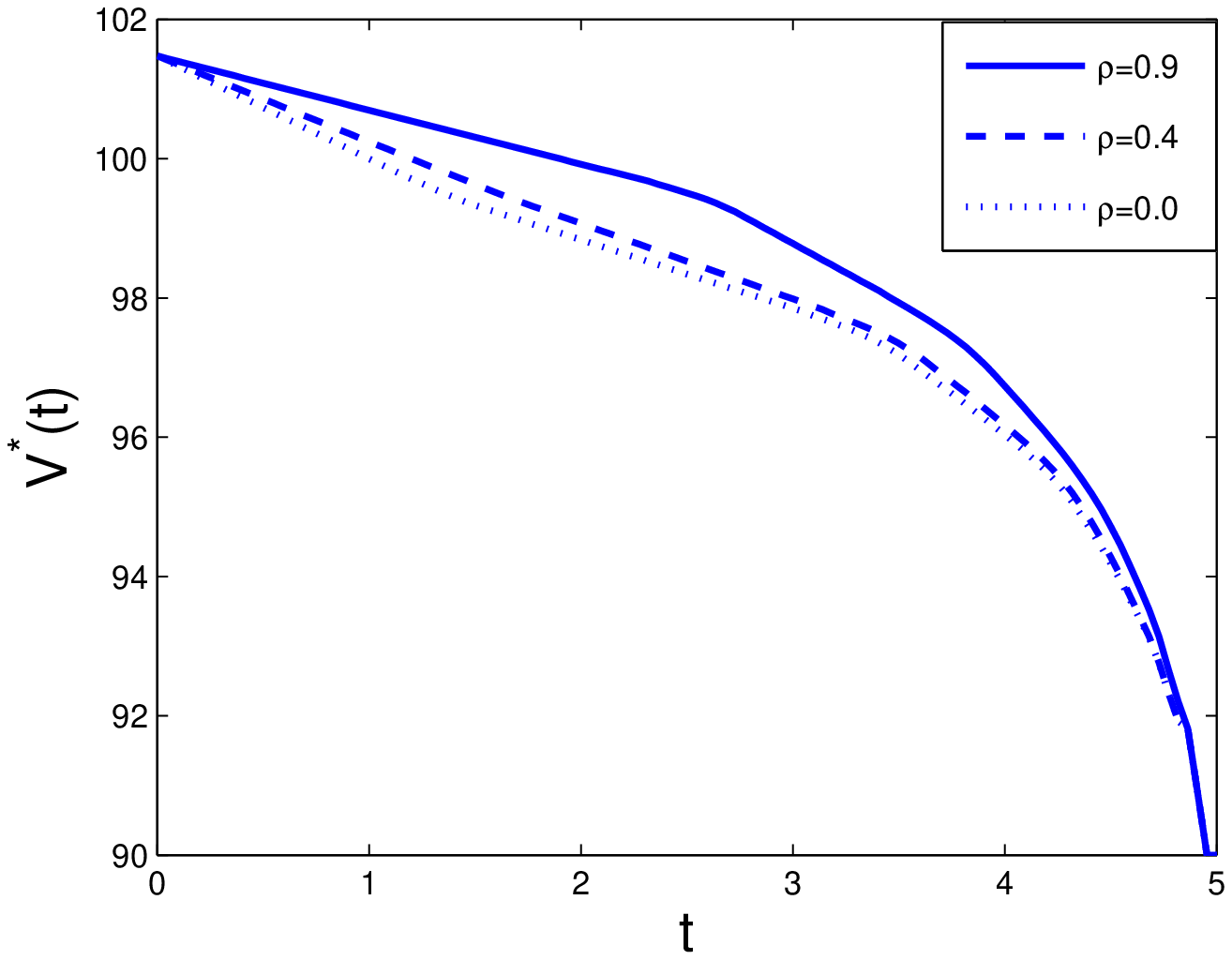}}
\subfloat{\includegraphics[scale=0.45]{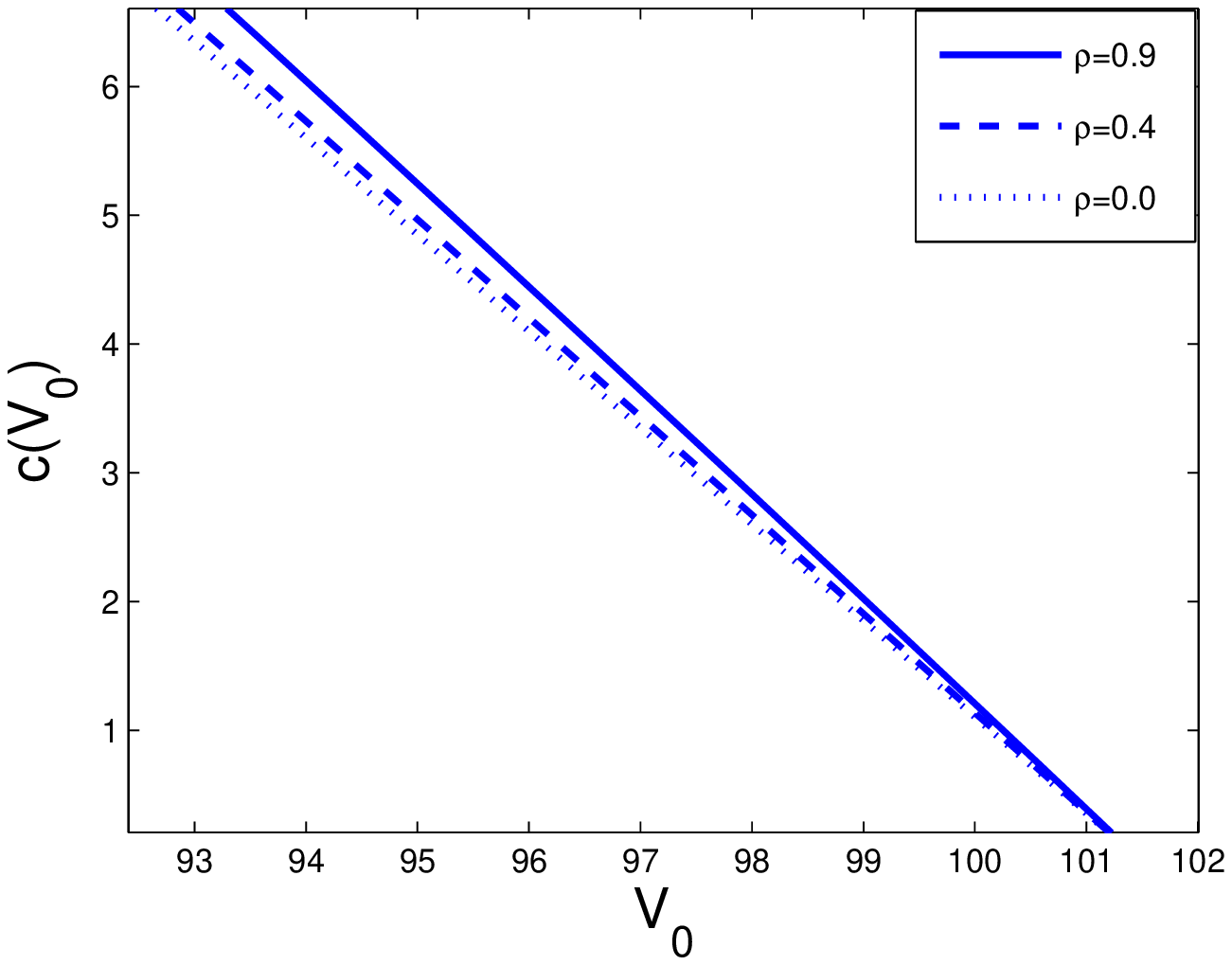}}\\
\subfloat{\includegraphics[scale=0.45]{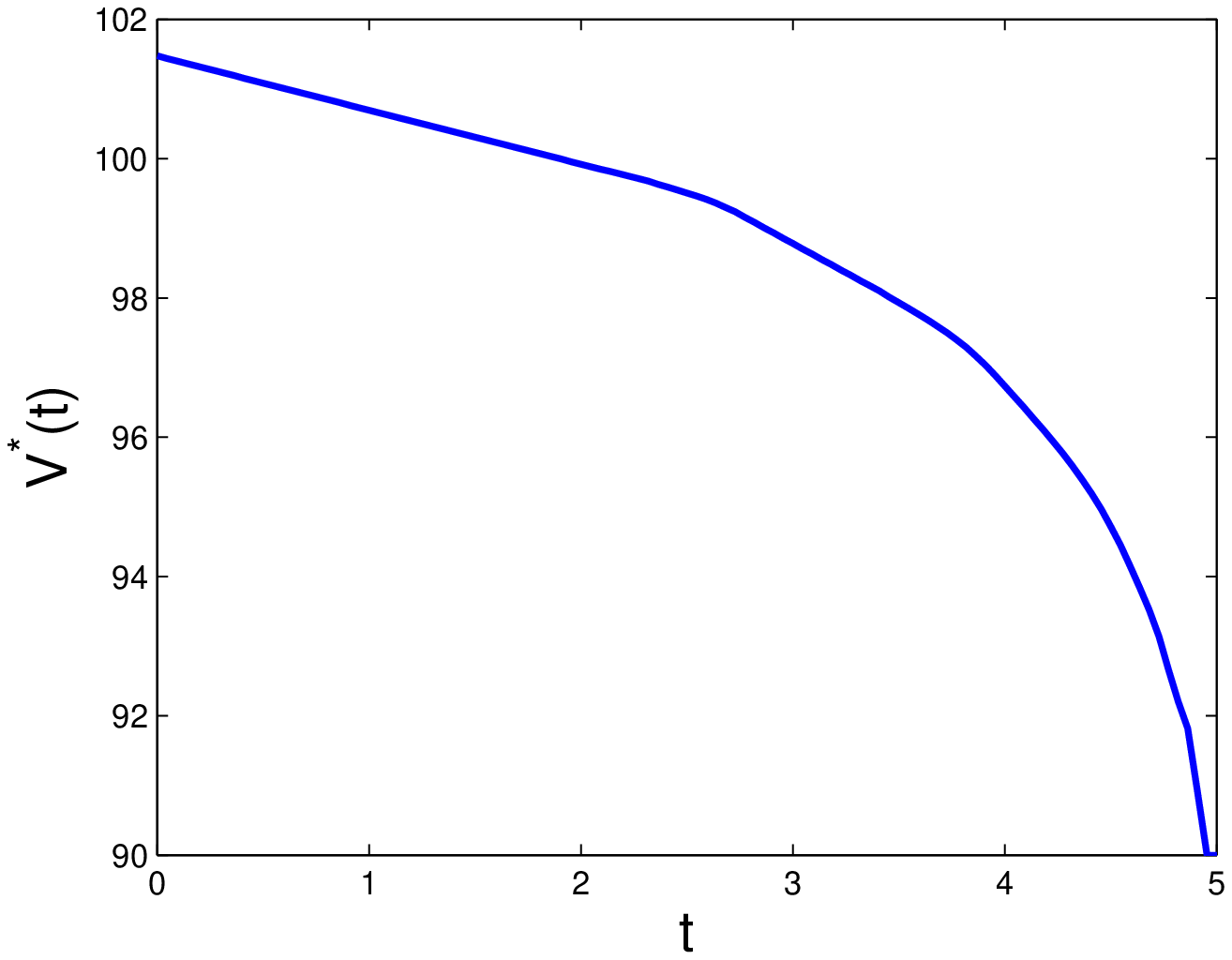}}
\subfloat{\includegraphics[scale=0.45]{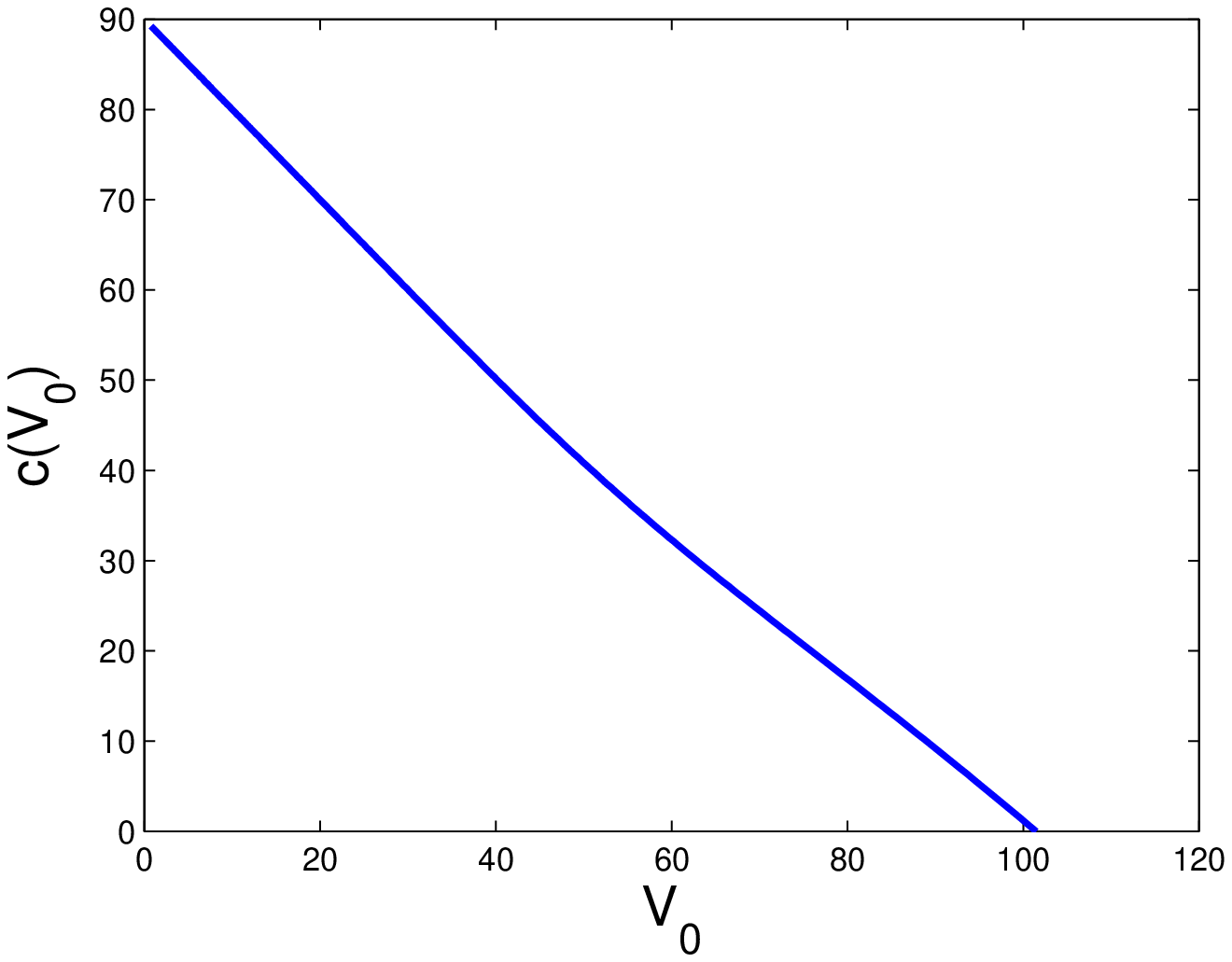}}
\caption{Dependence on model parameters for finite maturity}
\label{finite_fig}
\end{center}
\end{figure}

\section{Concluding remarks}

In this paper we have extended the analysis of \cite{XiaZhou07} for stock loans in incomplete markets. This allows us to consider the realistic situation when the borrower faces trading restrictions and cannot use replication arguments to find the unique arbitrage--free value for the repayment option embedded in such loans. We showed how an explicit expression for the loan fee can still be found in the infinite--horizon case provided the loan interest rate is set to be equal to the risk--free rate. In the finite--horizon case we characterize the loan fee in terms of a free--boundary problem and show how to calculate it numerically. In both cases, we analyzed how the loan fee depends on the underlying model parameters. 

Based on the dependence on correlation and risk--aversion, we find that the complete--market, risk--neutral valuation of a stock loan provides an upper bound for the fee to be charged by the bank. This shows that by following our model a bank can quantify the effects of the restrictions faced by the client thereby charging a smaller fee for the loan, presumably increasing its competitiveness.


\end{document}